% ...

\documentclass[aps,amssymb,amsmath,amsfonts,superscriptaddress]{revtex4}

\usepackage{graphicx}
\usepackage{bm,bbm}
\usepackage{epstopdf}
\usepackage{amsthm}
\usepackage{amsmath}
\usepackage{color}
\usepackage{graphicx}
\usepackage{graphics}
\usepackage{physics}
\usepackage[utf8]{inputenc}
\usepackage[OT4]{fontenc}
\usepackage{hyperref}

\newcommand{\1}{\mathbbm{1}} 

\usepackage[T1]{fontenc}
\usepackage[english]{babel}

\newtheorem{proposition}{Proposition}

\DeclareMathOperator{\id}{id}

\newcommand{\EE}{\mathcal{E}}

\begin{document}
\title{Pauli semigroups and unistochastic quantum channels}

\author{Zbigniew Pucha{\l}a}
\email{z.puchala@iitis.pl}
\affiliation{Institute of Theoretical and Applied Informatics, Polish Academy
of Sciences, ulica Ba{\l}tycka 5, 44-100 Gliwice, Poland}
\affiliation{Faculty of Physics, Astronomy and Applied Computer Science, 
Jagiellonian University, ul. {\L}ojasiewicza 11,  30-348 Krak{\'o}w, Poland}

\author{ \L ukasz Rudnicki} 
\affiliation{Max-Planck Institute for the Science of Light, Staudtstra{\ss}e 
2, 91058 Erlangen, Germany}
\affiliation{Center for Theoretical Physics, Polish Academy of Sciences, 
Al. Lotnik{\'o}w 32/46, 02-668 Warsaw, Poland}

\author{Karol {\.Z}yczkowski}
\affiliation{Faculty of Physics, Astronomy and Applied Computer Science, 
Jagiellonian University, ul. {\L}ojasiewicza 11,  30-348 Krak{\'o}w, Poland}
\affiliation{Center for Theoretical Physics, Polish Academy of Sciences, 
Al. Lotnik{\'o}w 32/46, 02-668 Warsaw, Poland}

\date{March 29, 2019}

\begin{abstract}
We adopt the perspective of similarity equivalence, in gate set tomography called the gauge, to analyze various properties of quantum operations belonging 
to a semigroup, $\Phi= e^{{\cal L}t}$,
and therefore given through the Lindblad operator. We first observe that the non unital part of the channel decouples from the time evolution. Focusing on unital operations we restrict our attention to the single-qubit case, showing that the semigroup embedded inside the tetrahedron of Pauli channels
is bounded by the surface composed of product probability vectors
and includes the identity map together with the maximally depolarizing channel.
Consequently, every member of the Pauli semigroup
is unitarily equivalent to a unistochastic map, describing a coupling with one-qubit environment
initially in the maximally mixed state, determined by a unitary matrix of order four.
\end{abstract}


\maketitle

\section{Introduction}

Quantum systems interacting with an environment
or being subject to a quantum noise can be described
within the theory of open quantum systems \cite{BP02,RH11}.
One applies the notion of mixed quantum state
represented by a positive, normalized hermitian matrix \cite{BZ17},
which can be considered as a generalization of the classical
probability vector.
 
In the stroboscopic approach, the dynamics is represented in discrete time steps
and the formalism of quantum operations, often  called \textit{quantum channels}, becomes useful.
These completely positive and trace preserving linear maps 
 send the set of mixed quantum states of a given size $N$ into itself.
  The channel can be considered as a generalization of the unitary evolution of a density matrix,  which takes into account interactions of the system with an environment or with a measurement apparatus. 
 
Although in the simplest case of a single qubit system, $N=2$,  the set of 
quantum operations has 12 dimensions, its structure and geometry
 is already well understood \cite{FA99,RSW02,BGNP14}.
 This  contrasts  the case of a larger size $N$, 
 for which the set of quantum operations has 
  $N^4-N^2$ dimensions and its structure and geometry
 become difficult to grasp  \cite{SWZ08}.

In an alternative approach, one describes dynamics in continuous time.
The most general form of a quantum Markov process,
which preserves positivity of quantum states,
is given by the equation derived by
Gorini, Kossakowski and  Sudarshan \cite{GKS76},
and independently by Lindblad \cite{Li76}.
Its solution describes a (non)unitary quantum dynamics, 
$\rho(t)=  e^{{\cal L}t} [\rho(0)]$,
and is determined by the Lindblad generator ${\cal L}$ -- see e.g. \cite{BP02}. 
During the last four decades  an approach of quantum dynamical semigroups 
\cite{AL98} was successfully used in a broad variety of physical problems. 
For an interesting account on the history and importance of the GKLS equation
 consult a recent review \cite{CP17}.

It is natural to ask to which extent both alternative descriptions of quantum
dynamics
are compatible. In general, this is not a trivial question as it is well known
that some quantum channels are not {\sl divisible} \cite{WC08,WECC08},
hence they cannot be represented as concatenation of other channels.
Consequently, they do not belong to a semigroup. In spite of several recent
related contributions 
on the structure of quantum Markovian dynamics \cite{FPMZ17,MCPS17}
the problem how to describe the set  of quantum channels,
which belong to a semigroup, remains open even in the simplest case of
operations acting on a single qubit \cite{RFZ10}. Nevertheless, as quite a lot
is known about formal properties of channels ranging from the divisible up to
the Markovian ones \cite{WC08,DZP18}, we are not concerned here with the
characterization problem. Instead, we assume that a given channel $\Phi$ can be
a \textit{seed}  for a semigroup (see next section for a proper explanation) and
we examine construction and properties of associated Lindblad operator.

Our approach pursued in this paper is inspired by gate set tomography \cite{Robin1,Robin2,Greenbaum} being an efficient successor of process tomography. 
In this reconstruction scheme, quantum channels are obtained up to a similarity equivalence, customarily called the gauge, i.e. instead of $\Phi$ one recovers $X\Phi X^{-1}$, where $X$ is unknown. It was shown \cite{RPZ18} that $X$ itself is a reversible and trace preserving operation, though, it is not necessarily a channel since it does not need to be completely positive. 

On the one hand side, some information concerning the character of the discrete
evolution corresponding to a quantum operation can be obtained by investigating
spectral properties of the corresponding superoperator \cite{CMM17,RPZ18}. On
the other hand, the gauge symmetry inherent to gate set tomography promotes the
spectrum of the superopertor $\Phi$ 
to be  the only source of accessible knowledge. Therefore, the aim of
this work is to make first steps in learning about the structure of the Lindblad
generators related to quantum channels, while being guided by the similarity
relation. In Sec. II we introduce necessary notation, while in Sec. III we
initiate the general analysis of the problem. In Sec. IV we concentrate on
the case $N=2$ and establish some results on the geometric structure of the set $\cal S$
of one-qubit unital operations which belong to a semigroup.
In particular, we demonstrate that any Pauli channel is unitarily equivalent
to a unistochastic channel.
 
 \section{Preliminaries}
 \label{sec_prem}
We consider a quantum channel $\EE$ acting on density matrices of order $N$. 
The action of the channel, defined in terms of the Kraus decomposition $K_j$
\begin{equation}
	\EE(\rho) = \sum_j K_j \rho K_j^\dagger,
	\label{kraus1}
\end{equation}
leads to the superoperator representation, $\Phi = \sum_j K_j \otimes {\bar 
K_j}$. Given any hermitian operator basis $\{B_0, \dots , B_{N^2-1} \}$, 
such that $B_0 =\1 / \sqrt{N}$ and $\tr B_i B_j = \delta_{ij}$, 
the  superoperator further acquires the block form
\begin{equation}
\Phi = \begin{bmatrix}
1 & 0 \\
\kappa & T
\end{bmatrix},
\label{PhiT}
\end{equation}
provided that we assume the channel $\EE$ is trace preserving. Up to unitary 
rotations in  $N^2 - 1$ dimensions, such a basis is formed by the generalized 
Pauli matrices. The real distortion matrix $T$ of order $N^2-1$ acts on the generalized 
Bloch vector representing a quantum mixed state in the Bloch representation, 
while the real vector $\kappa$ accounts for the displacement of the entire set 
of quantum states. In the case of unital maps, which preserve the maximally 
mixed state, $\Phi({\mathbbm 1}/N)= {\mathbbm 1}/N$, the translation vector 
vanishes,  $\kappa=0$. In the case of a unitary map, the matrix $T$ is 
orthogonal.

A matrix $\Phi$, acting on a composite Hilbert space ${\cal H}_N \otimes {\cal 
H}_N$,  can alternatively be represented in a product basis, $\ket{m} \otimes 
\ket{\mu}$ with matrix element written $ \Phi_{\stackrel{\scriptstyle m 
\mu}{n\nu}} = \bra{m \mu} \Phi \ket{n  \nu}.$  In general, the superoperator 
matrix $\Phi$ is non Hermitian. However, by reshuffling of its four indices
one obtains a Hermitian matrix $D_{\Phi}=\Phi^R$, where 
$X_{\stackrel{\scriptstyle m \mu}{n\nu}}^R =X_{\stackrel{\scriptstyle m n}{\mu 
\nu}}$  -- see~\cite{ZB04}. The matrix $D_{\Phi}$ is called dynamical matrix or 
the Choi matrix, as the theorem of Choi \cite{Ch75} implies that the map $\Phi$ 
is completely positive if and only if  the corresponding Choi matrix $D_{\Phi}$ 
is positive.

For every quantum operation  $\EE$  its corresponding superoperator $\Phi$ of 
order $N^2$ enjoys the following spectral properties \cite{BZ17}:
a) $\lambda_0 (\Phi) = 1$, due to trace preservation;
b) other $\lambda_i (\Phi), i = 1, . . . , N^2 - 1$ are in general complex, 
though they are either real or come in conjugate pairs, 
so that $\det \Phi$ and $\tr \Phi$ are real;
c)  All eigenvalues belong to the unit disk, $|\lambda_k| \leq 1$.

Although the spectrum of the superoperator can be complex,
the case $N=2$ is somewhat special:
For any one-qubit channel $\Phi$ there exists a unitarily equivalent operation
${\tilde \Phi}$ with real spectrum and diagonal distortion matrix,
\begin{equation}
\label{UV_equiv}
{\tilde \Phi}= \Psi_U \circ \Phi \circ \Psi_V=
(U\otimes {\bar U}) \Phi (V \otimes {\bar V}).
\end{equation}
To show this one uses a known group--theoretical  homomorphism, 
$SU(2) \simeq SO(3)$, which allows to represent a unitary transformation
of a complex one--qubit state by an orthogonal proper rotation of the corresponding
Bloch vector of length three  \cite{BGNP14}.
This implies the transformation $T\to {\tilde T}=O_U T O_V$,
where $O_U$ and $O_V$ denote orthogonal matrices of order three, 
which are determined by unitary rotation matrices of size two, $U$ and $V$ respectively.
An analogous formula with usage of arbitrary orthogonal matrices corresponds to singular
value decomposition of $T$ with positive singular values. 
Since in the case considered, orthogonal matrices do belong to $SO(3)$ and have unit determinant,
the transformed matrix is diagonal and real, but may contain also negative entries,
${\tilde T}={\rm diag}(\lambda_1, \lambda_2, \lambda_3)$.
The vector $\vec \lambda$ describes the distortion of the Bloch ball  
induced by the operation \cite{BZ17} and forms the spectrum of 
the operation ${\tilde \Phi}$.

\medskip
Let us now proceed to the description of quantum dynamics in continuous time. 
Every quantum Markov evolution can be determined by a linear Lindblad 
superoperator ${\cal L}$,
\begin{equation}\label{semig}
\rho(t)=  e^{{\cal L}t} [\rho(0)] =\Lambda_t [\rho(0)],
\end{equation}
which generates a semigroup, $\Lambda_s \Lambda_t = \Lambda_{t+s}$. According 
to the celebrated  GKLS theory \cite{GKS76,Li76}, the action of any Lindblad 
generator ${\cal L}$ can be written in terms of no more than $(N^2-1)$ jump 
operators $L_j$, 
\begin{equation}
{\cal L} (\rho)  = 
\sum_{j=1}^{N^2-1} 
\Bigl( L_j \rho L_j^{\dagger} -\frac{1}{2} L_j^{\dagger} L_j \rho
-\frac{1}{2} \rho L_j^{\dagger} L_j \Bigr).
\label{Lind2}
\end{equation}
A product of any three matrices, $Y=ABC$, can also be written as $Y=\Psi B$, 
where the superoperator reads $\Psi=A\otimes  C^T$ and $B$ is transformed to a 
vectorized form. Thus, the Lindblad generator can be explicitly represented 
by a matrix of order $N^2$,
\begin{equation} \label{Lind3}
	{\cal L} = 
	\sum_j L_j \otimes \overline{L}_j
	- \frac12 \sum_j L_j^\dagger L_j \otimes \1
	- \frac12 \sum_j \1 \otimes L_j^T  \overline{L}_j.
\end{equation}

Let us now merge both pictures by setting $\Phi=\Lambda_1$, so that ${\cal L}=\log\Phi$. We assume that $\Phi$ fulfills all the requirements to be a seed for the semigroup, what means that  $\Lambda_t=e^{t \log\Phi}$ gives the proper quantum channel for all $t\geq0$. Under a similarity transformation $\Phi \mapsto X\Phi X^{-1}$, we get
\begin{equation}
	{\cal L}\mapsto X \log \Phi X^{-1},\qquad \Lambda_t \mapsto Xe^{t \log\Phi} 
	X^{-1}.
\end{equation}
Clearly, the similarity transformation might spoil complete positivity of both the seed channel and the semigroup, however, it does not mingle with the time evolution.



\section{General channels}
First of all, we consider the general channel given by Eq. \eqref{PhiT} acting 
on a system of size $N$. We observe that, as long as the matrix $(T-\1)$ is 
invertible, the following decomposition holds:
\begin{equation}\label{Phikappa}
	\Phi=\left[\begin{array}{cc}
	1 & 0\\
	\kappa & T
	\end{array}\right]=\left[\begin{array}{cc}
	1 & 0\\
	\left(\1-T\right)^{-1}\kappa & \1
	\end{array}\right]\left[\begin{array}{cc}
	1 & 0\\
	0 & T
	\end{array}\right]\left[\begin{array}{cc}
	1 & 0\\
	\left(T-\1\right)^{-1}\kappa & \1
	\end{array}\right],
\end{equation}
and the last matrix on the right hand side is the inverse of the first one 
therein. Since Eq. \eqref{Phikappa} constitutes the similarity relation, we 
immediately conclude that the non unital contribution from $\kappa$ does not 
complicate the time evolution. Moreover, since the logarithm of the unital part 
in the middle term above is equal to  $\mathrm{diag}\left(0,\log T\right)$, the 
multiplication from the left will trivialize leading to the results:

\begin{equation}
	{\cal L}=\log \Phi =\left[\begin{array}{cc}
	1 & 0\\
	\left(\1-T\right)^{-1}\kappa & \1
	\end{array}\right]\left[\begin{array}{cc}
	0 & 0\\
	0 & \log T
	\end{array}\right]\left[\begin{array}{cc}
	1 & 0\\
	\left(T-\1\right)^{-1}\kappa & \1
	\end{array}\right]\equiv \left[\begin{array}{cc}
	0 & 0\\
	0 & \log T
	\end{array}\right]
	\begin{bmatrix}
	1 & 0 \\
	(T-\1)^{-1} \kappa & \1
	\end{bmatrix},
\end{equation}
and
\begin{equation}
	\Lambda_t = \left[\begin{array}{cc}
	1 & 0\\
	\left(\1-T\right)^{-1}\kappa & \1
	\end{array}\right]\left[\begin{array}{cc}
	1 & 0\\
	0 & T^t
	\end{array}\right]\left[\begin{array}{cc}
	1 & 0\\
	\left(T-\1\right)^{-1}\kappa & \1
	\end{array}\right]\equiv
	\begin{bmatrix}
	1 & 0 \\
	(T^t-\1)(T-\1)^{-1} \kappa & T^t
	\end{bmatrix}.
\end{equation}
The above simple and at the same time general formula, if $T$ is diagonal, could be intuitively 
rewritten in terms of the $N^2-1$ non-trivial eigenvalues $\lambda_i 
(\Phi)$, which for diagonal distortion matrices are real. If an \textit{a priori} given channel $\Phi = 
\Phi(\lambda_i,\kappa_i)$, with $i=1,\ldots,N^2-1$ is an admissible seed to a 
semigroup, then the following relation holds for an arbitrary dimension $N$ 
\begin{equation}
	\Lambda_t = 
	\Phi\left(\lambda_i^t, \frac{1-\lambda_i^t}{1-\lambda_i}\kappa_i\right).
	\label{kappan}
\end{equation}
Obviously, for $t=1$ this formula yields the seed quantum operation $\Phi$.

The general question, 
whether a given channel (\ref{kraus1}) defines a proper semigroup, 
is beyond the scope of this work.
Nevertheless, we shall briefly scrutinize few immediate restrictions. 
First of all, referring now to the Bloch form (\ref{PhiT}) of
the superoperator, we see that if all real eigenvalues of $\Phi$ are positive
the desired expression  $\log T$ can be defined.
%
 Moreover, existence of $(T-\1)^{-1}$ forces that all eigenvalues 
(except $\lambda_0$), if real, are strictly smaller than 1. For a qubit case, 
$N=2$, 
it is known that \cite{RPZ18}
\begin{equation}
\| \kappa \|^2 \leq 1 - |\lambda_1|^2 - |\lambda_2|^2 - |\lambda_3|^2 +
2 \lambda_2 \lambda_2 \lambda_3,
\end{equation}
is a necessary condition for complete positivity of $\Phi$. Note that it stems from a related condition expressed through singular values of $T$ \cite{BGNP14}.
From this result we easily infer, that if $ \kappa \neq 0$, then
\begin{equation}\label{eqn:max-lambda}
\max\{|\lambda_1|,|\lambda_2|,|\lambda_3|\}  < 1.
\end{equation}
so that the second requirement is always satisfied. To prove the assertion, 
without loss of generality, we assume that max 
$\{|\lambda_1| , |\lambda_2| , |\lambda_3|\} = |\lambda_1|$.
Then if $|\lambda_1| = 1$ we obtain
\begin{equation}
\|\kappa\|^2 \leq - |\lambda_2|^2 - |\lambda_3|^2 + 2 |\lambda_2 \lambda_3| \ \leq \ 0,
\end{equation}
which enforces the contradiction $\kappa=0$. The last inequality used is arithmetic-geometric. A similar property could potentially hold for $N>2$, however, further studies are needed to substantiate that hope.

\section{Qubit channels}
In the second part we consider the simplest but fairly non-trivial one-qubit problem, $N = 2$. 
In this case the three eigenvalues of $T$ are either all real or of the form: 
$\lambda_1 (\Phi) = x \in \mathbb{R}$ and $\lambda_2(\Phi) = z, \lambda_3(\Phi) 
= {\bar z} $ with $z \in \mathbb{C}$. In what follows we concentrate on the 
first case, leaving the second one as an open problem for the future. 

Above, we have already shown that every channel (also beyond the qubit case) is 
similar to its own unital variant, obtained by letting $\kappa\rightarrow0$. 
Moreover, every unital qubit channel with four real eigenvalues 
($\lambda_0\equiv 1,\lambda_1, \lambda_2, \lambda_3$) is similar (See Eq. (40) in
 ref. \cite{RPZ18}) to the channel
\begin{equation}\label{Xi}
	\Xi= \sum_{i=0}^3 p_i \sigma_i \otimes {\bar \sigma_i}
	=\left(
	\begin{smallmatrix}
	p_0 + p_3 & 0  & 0  &  p_1 + p_2\\
	0 & p_0-p_3  & p_1 - p_2  & 0\\
	0 & p_1-p_2  & p_0 - p_3  & 0\\
	p_1 + p_2 & 0  & 0  & p_0 + p_3
	\end{smallmatrix}
	\right).
\end{equation}
Note that, for further convenience, $\Xi$ is given in the product basis (not in the Pauli basis) so it is not of the form (\ref{PhiT}). However, for the sake of the spectrum, the choice of the basis is of no relevance. The eigenvalues are related to the probabilities $p_0,\ldots,p_3$ in the following way:
\begin{equation} \label{lambp}
	\begin{split}
	\lambda_0 = p_0+p_1+p_2+p_3 &\equiv 1, \\
	\lambda_1 = p_0+p_1-p_2-p_3 &= 1-2(p_2+p_3), \\
	\lambda_2 = p_0-p_1+p_2-p_3 &= 1-2(p_1+p_3), \\
	\lambda_3 = p_0-p_1-p_2+p_3 &= 1-2(p_1+p_2).
	\end{split}
\end{equation}

In other words, every single-qubit unital map is similar to a Pauli channel  
--- mixed unitary operation, defined as convex combination of rotations with 
respect to the Pauli matrices, $\Xi : \rho \mapsto \sum_{i=0}^3  p_i \sigma_i 
\rho \sigma_i^\dagger$, with $\sigma_0 = \id, \sigma_1 = \sigma_x, \sigma_2 = 
\sigma_y, \sigma_3 = \sigma_z$. For such a channel the Kraus operators can be 
chosen as $K_i = \sqrt{p_i} \sigma_i$. In fact, taking into account the results 
presented in the previous section, every qubit channel is similar to $\Xi$, 
with the similarity transformation $X$ being given by concatenation of  the 
transformation from Eq. (\ref{Phikappa}) and the rotation which brings the 
intermediate unital channel to the form (\ref{Xi}).

Quite naturally, every channel (non unital and of any dimension) is also 
similar to  ${\rm diag}(1,\ \lambda_1,\ldots,\lambda_{N^2-1})$, so that from 
this fundamental perspective, the time evolution does only depend on the 
eigenvalues. However, the observations made so far point towards the less 
obvious approach, in which for qubits, the burden and actual complexity are 
both to be lifted to the level of the Pauli channels. Therefore, in the next 
four subsections we discuss various semigroup-related properties of this 
special class of channels.

\subsection{The Pauli semigroup}
It is easy to check that the probabilities defining the Pauli channel are 
non-negative whenever $\lambda^{\downarrow}_{3} \geq \lambda^{\downarrow}_{1} + 
\lambda^{\downarrow}_2 -1$, where $\lambda^{\downarrow}$ denotes the vector 
$(\lambda_1,\lambda_2,\lambda_3)$ ordered decreasingly. In fact, this is a 
necessary and sufficient condition for complete positivity of this channel, 
derived in terms of the eigenvalues \cite{RPZ18}. Moreover, if we employ the 
condition $\lambda_{\min} \geq 0$, which holds whenever $p_i + p_j \leq 
\frac12$ for $i,j\in \{1,2,3\}; i \neq j$, the logarithm of $\Xi $ is 
well-defined. These conditions imply that $p_0$ is the largest component of 
vector $p$.

Let us introduce an orthogonal matrix $O_4$, an explicit form of which is of no 
relevance here, that allows to diagonalize the superoperator in question,
 $E=O_4 \Xi O_4^\top  = {\rm diag}(1,\lambda_1,\lambda_2,\lambda_3)$. Under all the conditions listed above, the Lindblad generator is
\begin{equation}
	{\cal L} = O_4^\top \log E \ O_4.
\end{equation}
The associated semigroup is then given by
\begin{equation}
	\Lambda_t = e^{{\cal L} t} =  O_4^\top e^{t \log E}  O_4 = 
	\frac12 \left(
	\begin{smallmatrix}
	1+\lambda_3^t & 0 & 0 & 1-\lambda_3^t\\
	0 & \lambda_1^t+\lambda_2^t & \lambda_1^t-\lambda_2^t & 0\\
	0 & \lambda_1^t-\lambda_2^t & \lambda_1^t+\lambda_2^t & 0\\
	1-\lambda_3^t & 0 & 0 & 1+\lambda_3^t
	\end{smallmatrix}
	\right).
\end{equation}
In order to check if this is a completely positive operation for all $t\geq 0$, 
we construct the corresponding Choi matrix
\begin{equation}
	D_{\Lambda_t}=\Lambda_t^R =
	\frac12 \left(
	\begin{smallmatrix}
	1+\lambda_3^t & 0 & 0 & \lambda_1^t+\lambda_2^t\\
	0 & 1-\lambda_3^t & \lambda_1^t-\lambda_2^t & 0\\
	0 & \lambda_1^t-\lambda_2^t &  1-\lambda_3^t& 0\\
	\lambda_1^t+\lambda_2^t & 0 & 0 & 1+\lambda_3^t
	\end{smallmatrix}
	\right),
\end{equation}
and study its positivity. As recalled in Section \ref{sec_prem},
for any four-index matrix $Y$ the symbol
 $Y^R$ denotes the matrix with reshuffled indices \cite{ZB04};
such an  involution 
transforms a non-hermitian superoperator matrix $\Phi$ into 
the hermitian dynamical 
matrix $D_\Phi$. The above matrix is positive iff the following three relations 
are satisfied for all times $t\ge 0$:
\begin{equation}
	\begin{split}
	1+\lambda_3^t &\geq \lambda_1^t+\lambda_2^t, \\
	1+\lambda_2^t &\geq \lambda_1^t+\lambda_3^t, \\
	1+\lambda_1^t &\geq \lambda_2^t+\lambda_3^t. 
	\end{split}
\end{equation}
Expanding these equations in power series around $t=0$ we see that it is 
sufficient to check the inequalities for small $t>0$ -- see also  Eq. (51) in 
the work of  Wolf and Cirac \cite{WC08}.
This gives three independent conditions  on the eigenvalues of the Bloch 
transition matrix $T$ entering Eq. (\ref{PhiT}), which need to be fulfilled to 
assure positivity of the Choi matrix:
\begin{equation} 
\label{lambda}
	\lambda_3 \geq \lambda_1 \lambda_2, \ \ \ \
	\lambda_2 \geq \lambda_1 \lambda_3, \ \ \ \
	\lambda_1 \geq \lambda_2 \lambda_3.
\end{equation}


\subsection{The geometric picture}
As we shall see below, each of the conditions (\ref{lambda}) determines a 
surface inside the tetrahedron of the Pauli channels and forms a part of the 
boundary of the  set $\cal S$ of operations belonging to a semigroup. For any 
point outside the set $\cal S$ one can try to find a continuous trajectory 
starting at identity, but for some finite time $t$ it will leave the 
tetrahedron of completely positive maps --- see Fig. 1.

\begin{figure}[h!]
\centering
\includegraphics[width=0.41\linewidth]{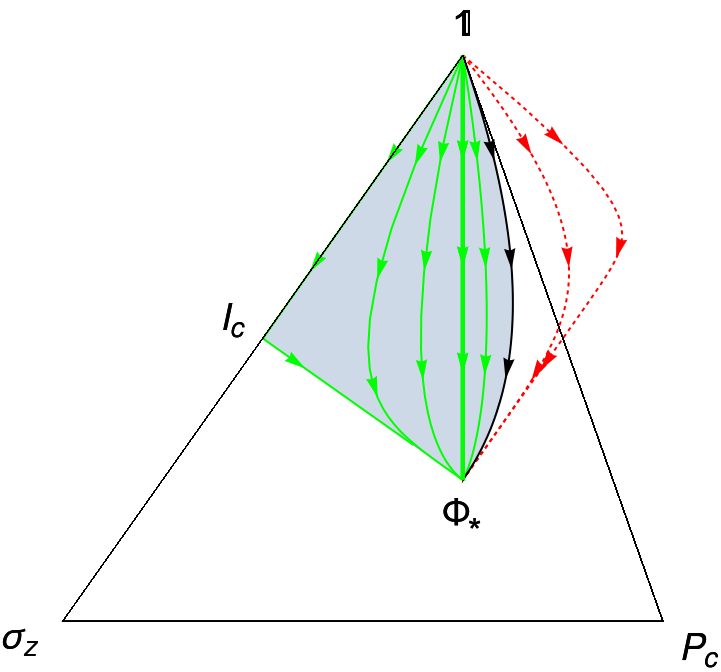}
\caption{Cross-section of the simplex of one--qubit Pauli channels 
 determined by the identity map $\mathbbm 1$, the completely depolarizing channel  
$\Phi_*$, and the $z$-rotation corresponding to the map 
$\Phi_z=\sigma_z\otimes\sigma_z$
which includes also classical channels:
$I_c=({\mathbbm 1}+\Phi_z)/2={\rm diag}(1,0,0,1)$ and
$P_c=(\Phi_x+\Phi_y)/2={\rm diag}(0,1,1,0)$.
The set $\cal S$ of channels belonging to a semigroup
is shown in gray, and its right boundary corresponds to the product relation,
$p_0p_3=p_1p_2$.  
Some exemplary semigroups leading from $\mathbbm 1$ to $\Phi_*$, 
 are represented by solid arrowed lines. Dashed (red) lines 
 going through a channel $\Psi \notin {\cal S}$
does not correspond to a semigroup,  as 
 it leaves the simplex of CP maps for some initial time $t>0$.}
\label{fig:channels-section}
\end{figure}

The relations (\ref{lambp}) allow us to translate the inequalities 
(\ref{lambda}) concerning eigenvalues $\lambda_i$ into constraints for 
components of the probability vector:
\begin{equation}\label{eqn:probab-bounds}
	p_0 p_3 \geq  p_1 p_2, \ \ \ \
	p_0 p_2 \geq  p_1 p_3, \ \ \ \
	p_0 p_1 \geq  p_2 p_3.
\end{equation}
Boundary of the set $\cal S$ is met, whenever one of these inequalities becomes 
saturated. This happens if the classical probability vector of length four
has a product structure. For example, the choice
\begin{equation}\label{prod}
	{\vec p}=(a,a') \times (b,b')=(ab,ab',a'b,a'b'),
\end{equation}
where $a'=1-a$ and $b'=1-b$, renders $p_0p_3=aba'b'=p_1p_2$. Other two equalities would correspond to permutations of the components of ${\vec p}$.

Interestingly, these product vectors form three fragments of a hyperboloid --
a ruled surface inside the tetrahedron which is spanned by its two opposite 
sides. This manifold  is useful to visualize the set of  two-qubit separable 
states in a 3D cross-section through the 6D space of two-qubit pure states
\cite{BBZ02,BZ17} and to identify the maximally entangled states, which are 
located as far from this manifold as possible. 

The region  $\cal S$ for which constraints (\ref{eqn:probab-bounds}) are 
satisfied contains two special points: $(1,0,0,0)$, corresponding to identity 
channel, and $(1,1,1,1)/4$, representing the maximally depolarizing channel 
$\Phi_*$. The region in question is also bounded by three surfaces consisting 
of product probability vectors --- see 
Fig.~\ref{fig:channels-eigenvalues-semigroup}. Below we prove a proposition concerned with the shape of $S$, which stems from the product structure mentioned in the previous sentence.


\begin{figure}[h!]
\centering
\includegraphics[width=0.47\linewidth]{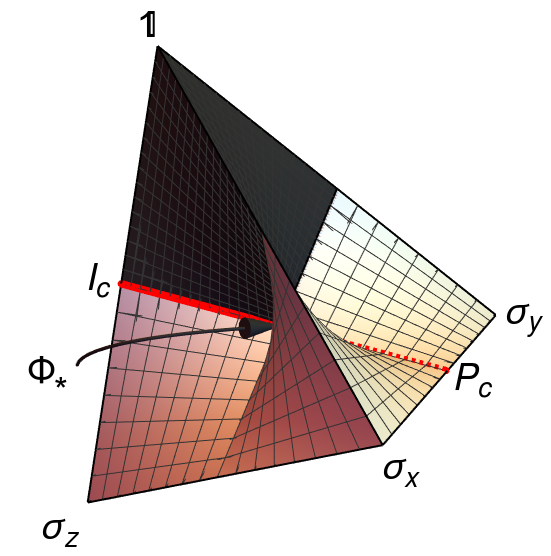}
\caption{
	The tetrahedron of  Pauli channels spanned by three Pauli matrices and 
	identity contains a proper set of maps bounded by the classical 4-point 
	probability vectors with a product structure. Dark region, including 
	identity map $\mathbbm 1$ and the completely depolarizing channel  
	$\Phi_*$, represents probability vectors for which the corresponding Pauli 
	channel $\Phi_{\vec p}$ belongs to a semigroup. Dashed (red) line 
	represents maps corresponding to classical action of bistochastic matrices 
	of size $N=2$, while solid line denotes maps from the classical semigroup 
	$[I_c, \Phi_*]$. The map $P_c$ represents negation of the classical state. 	
}
\label{fig:channels-eigenvalues-semigroup}
\end{figure}

\begin{proposition}
The set $\cal S$ is star-shaped with respect to any point on the interval 
\begin{equation}\label{eqn:star-shaped-interval}
\beta  (1,0,0,0) + (1-\beta)  \frac14(1,1,1,1), \ \text{ for } \ \beta \in 
[0,1].
\end{equation}
\end{proposition}
\begin{proof}
	We will show, that the intervals connecting any boundary point of ${\cal 
	S}$ with an arbitrary point from the interval~\eqref{eqn:star-shaped-interval},
	represented by a bold interval in Fig. 1,  belongs to 
	${\cal S}$. 	
%
	The boundary of the set ${\cal S}$ is solely formed by permutations of 
	product vectors, so to fix the attention we consider a part of the boundary given 
	by $(p q, p(1-q), (1-q) p, (1-p)(1-q))$, with $\frac12 \leq 
	p,q \leq 1$. For every point from the 
	interval~\eqref{eqn:star-shaped-interval} we consider a line to the above generic boundary point, which is parameterized as
	\begin{equation}
		\vec{r}=\gamma [\beta  (1,0,0,0) + (1-\beta)  \frac14(1,1,1,1)]
		+ (1-\gamma) (p q, p(1-q), (1-q) p, (1-p)(1-q)) \ \text{ for } \ 
		\beta,\gamma\in[0,1].
	\end{equation}
	With the help of elementary inequalities one can check that under our assumptions the condition $r_0 
	r_3 \geq r_1 r_2$ holds. Consequently, $\vec{r} \in {\cal S}$ for all $\beta,\gamma\in[0,1]$.
\end{proof}

\subsection{Lindblad dynamics associated with Pauli semigroups}
 
Consider now a special example of the Pauli semigroup 
$\Lambda_s^z=\exp({\cal L}_z s)$ associated with a single jump operator
 $L_1=\sigma_z$. It describes the effects of decoherence
  as it gradually diminishes the off-diagonal entries of 
the density matrix. Define also two accompanying semigroups: 
$\Lambda_t^x=\exp({\cal L}_x t)$  describing decoherence with respect to the
basis along the  $x$ axis and   corresponding 
to the jump operator $L_1=\sigma_x$, The third Pauli semigroup, 
$\Lambda_u^y=\exp({\cal L}_y u)$ is defined analogously. Since the tensor squares 
of the Pauli matrices commute, $[\sigma_i\otimes \sigma_i, \sigma_j \otimes 
\sigma_j]=0$, the corresponding Lindblad generators commute as well, $[{\cal 
L}_i, {\cal L}_j]=0$ for $i,j=x,y,z$, and so do the semigroups. Therefore, the 
order of operations is not important and will not influence the following key 
observation concerning  a composition of two or three semigroups:
 
\medskip
\begin{proposition}
	For any choice of the times $(t,s)$ the map $\Lambda_s^z\Lambda_t^x$
	gives a Pauli channel for which the vector $\vec{p}$ has a product 
	structure (\ref{prod}). This map forms a part of the boundary of the set 
	$\cal S$. Two other parts of this boundary are obtained by the remaining 
	two-factor concatenations $\Lambda_s^z\Lambda_u^y$ and 
	$\Lambda_u^y\Lambda_t^x$.
\end{proposition}

\begin{proof}
	Using the fact that tensor squares of Pauli matrices do commute, we arrive 
	at a simple result
	\begin{equation}
		\Lambda_s^z  \Lambda_t^x = O_4^\top \mathrm{diag}(1,e^{-2 s}, 
		e^{-2 (s+t)}, e^{-2 t})O_4. 
	\end{equation}
	Therefore $\Lambda_s^z \Lambda_t^x$ is a Pauli channel for which the 
	corresponding probability vector:
	\begin{equation}
		\begin{split}
		p_0 &=  2 e^{-(s+t)}  \cosh (s) \cosh (t), \\
		p_1 &=  2 e^{-(s+t)}  \cosh (s) \sinh (t), \\
		p_2 &=  2 e^{-(s+t)}  \sinh (s) \sinh (t), \\
		p_3 &=  2 e^{-(s+t)}  \sinh (s) \cosh (t),
		\end{split}
	\end{equation}
	 enjoys the product structure $p_0p_2=p_1p_3$.
	 Other two conditions  for a product vector $p$  correspond to remaining 
	  choices of two generators and  two other parts of the boundary of  $\cal S$.
\end{proof}

\begin{proposition}
	For any point $\vec{p}$ in the interior of the set $\cal S$ of quantum 
	operations belonging to the semigroup, 	there exist a triple ($s,t,u>0$)
	such that the corresponding Pauli channel is given by a composition 	
	$\Phi_{\vec{p}} = \Lambda_s^z\Lambda_t^x\Lambda_u^y$.
\end{proposition}

\begin{proof}
	For a given probability vector $\vec{p} \in \cal S$	we resort to \eqref{lambp} in order to find its accompanying distortion vector $\vec 
	\lambda$. In the next step we set:
	\begin{equation}
		\begin{split}
		s = \frac12 \log \sqrt{\frac{\lambda_3}{\lambda_1 \lambda_2}},\\
		t = \frac12 \log \sqrt{\frac{\lambda_1}{\lambda_2 \lambda_3}},\\
		u = \frac12 \log \sqrt{\frac{\lambda_2}{\lambda_1 \lambda_3}}.
		\end{split}
	\end{equation}
	For points in the interior, relations (\ref{lambda}) are satisfied as 
	strict inequalities, the arguments of the logarithms 
	are greater than unity so that the times $s,t,u$ are positive.
	
	Now utilizing the trivial commutation relation, we get 
	\begin{equation}
		\Lambda_s^z \Lambda_t^x  \Lambda_u^y = 
		O_4^\top \mathrm{diag}(1,e^{-2 (s+u)}, e^{-2 (s+t)}, e^{-2 (t + u)})O_4,
	\end{equation}
	which gives $\Phi_{\vec{p}}$. Note that the entries of the diagonal matrix
	reveal time dependece of the vector $\lambda$.
\end{proof}

\subsection{Dynamical semigroups and unistochastic maps}

Let us recall here that a unital quantum operation acting on an $N$ dimensional 
system and determined by a unitary matrix  $U$ of order $N^2$, which describes 
the coupling of the system with environment initially in the maximally mixed 
state followed by the partial trace over the environment $B$, 
\begin{equation}\label{unist}
	\rho' =\Phi_U(\rho) =  \tr_B \Bigl( U(\rho \otimes \frac{{\mathbbm I}_N}{N})
	U^\dagger \Bigr),
\end{equation}
is called {\sl unistochastic} \cite{ZB04}. It is known that the set of 
one-qubit unistochastic maps forms a non-convex proper subset of the 
tetrahedron of Pauli channels~\cite{MKZ13} which is bounded by the manifold of 
probability vectors with a tensor product structure as shown in Fig. 1.  
Therefore, the set of Pauli channels belonging to a semigroup forms a quarter 
of the set of unistochastic maps which contains the identity operation. More 
formally, we arrive at the following statement.

\begin{proposition} \label{prop:unitochastic-in-semigroup}
	Any bistochastic one qubit Pauli channel $\Phi : \rho \mapsto \sum_{i=0}^3 
	p_i \sigma_i \rho \sigma_i^\dagger$ belongs to a dynamical semigroup, and 
	can be written in the Lindblad form, $\Lambda_t = e^{\mathcal{L} t}$, if and only if 
	$\Lambda_t$ is unistochastic and the identity component $p_0$ is the 
	largest component of the probability vector $\vec{p} = (p_0,p_1,p_2,p_3)$.
\end{proposition}

To shed some light on the above result we will make use of the notion of 
locally equivalent gates and apply the canonical form of a two-qubit quantum 
gate. Two unitary matrices $U$ and $V$ of order four, are called {\sl locally 
equivalent}, if there exist four unitary matrices $W_i$  of order two such that
\begin{equation}\label{eqn:unitary-equivalence}
	V\sim U=(W_1\otimes W_2)  V (W_3 \otimes W_4).
\end{equation}  
Observe that  unistochastic maps $\Phi_U$ and $\Phi_V$, generated by two locally equivalent 
unitary matrices, $U \sim V$, lead to unitarily equivalent channels
(\ref{UV_equiv}).

Any unitary matrix $U\in U(4)$  is locally equivalent to a two--qubit gate 
written in the canonical Cartan form~\cite{KBG01,KC01},
\begin{equation}
\label{cartan}
	U \sim V = e^{i \sum_{j=1}^3 \alpha_j \sigma_j  \otimes \sigma_j},
\end{equation} 
and the vector of $(\alpha_1,\alpha_2, \alpha_3)$ called {\sl information 
content} can be chosen from the Weyl chamber, $\frac{\pi}{4}\geq \alpha_1 \geq 
\alpha_2\geq |\alpha_3 |\geq~0$.
 Thus, we are going to restrict our attention to the unistochastic maps $\Phi_U$
 corresponding to unitary matrices in the above
  Cartan form. 

Consider a unistochastic channel $\Phi_U$  determined by a unitary $U\in U(4)$. 
The corresponding Choi matrix can be written with use of the reshuffled matrix,
$D_U = \frac12 U^R (U^R)^{\dagger}$, so that the superoperator 
reads~\cite{MKZ13}, $\Phi_U = \frac12 [U^R (U^R)^{\dagger}]^R$. Taking $U$ in 
the form Eq. (\ref{cartan}) we arrive at the following expression for the 
superoperator 
\begin{equation}
	\Phi_U = O_4^\top	
	\left(
	\begin{smallmatrix}
	1 & 0 & 0 & 0 \\
	0 & \cos \left(2 \alpha _2\right) \cos \left(2 \alpha _3\right) & 0 & 0 \\
	0 & 0 & \cos \left(2 \alpha _1\right) \cos \left(2 \alpha _3\right) & 0 \\
	0 & 0 & 0 & \cos \left(2 \alpha _1\right) \cos \left(2 \alpha _2\right) \\
	\end{smallmatrix}
	\right) O_4.
\end{equation}
Since the eigenvalues of $\Phi_U$ do satisfy~\eqref{lambda} this channel 
belongs to a semigroup. On the other hand for a qubit unital channel with 
positive eigenvalues, satisfying constraints~\eqref{lambda}, we can find 
unitary matrix $U$, such that the channel can be written as $\Phi_U$. It is 
enough to take the angles $\alpha_1,\alpha_2,\alpha_3$
\begin{equation}
	\begin{split}
	\alpha_1 = \frac12 \arccos(\sqrt{\frac{\lambda_2\lambda_3}{\lambda_1}}),\\
	\alpha_2 = \frac12 \arccos(\sqrt{\frac{\lambda_1\lambda_3}{\lambda_2}}),\\
	\alpha_3 = \frac12 \arccos(\sqrt{\frac{\lambda_1\lambda_2}{\lambda_3}}),
	\end{split}
\end{equation}
as due to relations (\ref{lambda}) the arguments of the square roots do not 
exceed the unity. In this way we have shown that every unistochastic channel 
(up to a unitary rotation) belongs to a semigroup. This concludes a constructive 
proof of Proposition~\ref{prop:unitochastic-in-semigroup}.

\section{Concluding remarks}

The results of the present work, motivated by the similarity equivalence (gauge) inherent to gate set tomography, touched upon the general case of non unital quantum maps of arbitrary dimension and studied in more detail the set ${\cal S}$ of single-qubit quantum Pauli channels which belong to a semigroup. A conceptual novelty of the results is that the Pauli channels are shown to be a handy way of encoding relevant features of qubit channels.

To synthetically summarize the findings of geometrical nature: the set ${\cal S}$ forms a subset of a quarter of 
the tetrahedron of Pauli channels bounded by the surfaces consisting of 4-point
probability vectors with the tensor product structure. 
Altough this set  is not convex, it is star-shaped. Since this set contains 
unistochastic channels~\cite{MKZ13}, which correspond to the coupling with a 
one-qubit environment initially in the maximally mixed state~\cite{ZB04}, we  conclude that every map accessible through the  continuous semigroup is 
unitarily equivalent to a unistochastic channel. 

Note that for any quantum operation acting on quantum states of size $N$
one can find the corresponding classical map represented by a stochastic 
transition matrix of order $N$. This matrix is obtained by reshaping the 
diagonal of the  Choi matrix, which may be interpreted as a result of the 
decoherence acting in the space of quantum maps \cite{KCPZ18,KCPZ19}.
Hence the problem studied in this work can be considered as a quantum analogue 
of the question,  which bistochastic matrix allows for a continuous dynamics, 
$B=\exp({\cal L} t)$, such that the trajectory does not leave the Birkhoff 
polytope ${\cal B}_N$ of bistochastic matrices \cite{BC16}. 
Although  for $N=2$ this 
question is fairly easy, [the answer is $B_a=(1-a,a;\ a,1-a)$ with $a\le 1/2$ 
represented by interval $[I_c, \Phi_*]$ in Fig. 2], already for $N=3$ the 
problem becomes interesting \cite{HL11,SZ18}. 
In fact one can study also a simpler 
question, asking for the set of bistochastic matrices, such that its square 
root (or roots of a higher order) are bistochastic and formulate the 
corresponding quantum problem,  looking for bistochastic operations such that 
its square root  (or higher order roots) forms a bistochastic operations 
\cite{HL11,Sn14}. 
Any divisible map, which belongs to a semigroup, is included 
inside this set.

It will be a challenge to find out, which of the above results can be extended 
to a more general class of systems. In particular, in the case of unital quantum channels 
 it is clear that the set of the bistochastic maps acting on ${\cal 
H}_N$, which belong to a semigroup, includes the maps corresponding to the 
action of bistochastic matrices of size $N$, which belong to classical 
semigroups ~\cite{BC16,SZ18}.

\medskip

We are grateful to  Fabio Benatti, Dariusz Chru{\'s}ci{\'n}ski and 
Fereshte Shahbeigi for numerous stimulating discussions on divisibility and quantum semigroups.  
It is a great pleasure to thank David Davalos and Carlos Pineda for fruitful 
discussions, which allowed us to understand similarities and differences between
both approaches to the problem.  K{\.Z} is grateful to Chryssomalis 
Chryssomalakos for the invitation to Mexico, which made this interaction 
possible. Financial support by the Polish National Science Centre (NCN) under 
the grant numbers DEC-2015/18/A/ST2/00274 (K.{\.Z}); 2016/22/E/ST6/00062 
(Z.P.) is gratefully acknowledged.

\end{document}